\documentclass[a4paper,cleveref,thm-restate,authorcolumns]{lipics-v2021}

\usepackage{cite}
\usepackage{xspace}

\crefname{rule}{rule}{rules}

\title{Sinkless Orientation Made Simple}
\author{Alkida Balliu}{Gran Sasso Science Institute, Italy}{alkida.balliu@gssi.it}{https://orcid.org/0000-0001-5293-8365}{}
\author{Janne H.\ Korhonen}{IST Austria, Austria}{janne.h.korhonen@gmail.com}{}{European Research Council (ERC) under the European Union’s Horizon 2020 research and innovation programme (grant agreement No 805223 ScaleML).}
\author{Fabian Kuhn}{University of Freiburg, Germany}{kuhn@cs.uni-freiburg.de}{}{}
\author{Henrik Lievonen}{Aalto University, Finland}{henrik.lievonen@aalto.fi}{https://orcid.org/0000-0002-1136-522X}{Academy of Finland (grant agreement No 333837).}
\author{Dennis Olivetti}{Gran Sasso Science Institute, Italy}{dennis.olivetti@gssi.it}{}{}
\author{Shreyas Pai}{Aalto University, Finland}{shreyas.pai@aalto.fi}{https://orcid.org/0000-0003-2409-7807}{}
\author{Ami Paz}{LISN, CNRS, France}{ami.paz@lisn.fr}{https://orcid.org/0000-0002-6629-8335}{Austrian Science Fund (FWF) and netIDEE (grant agreement No P~33775-N).}
\author{Joel Rybicki}{IST Austria, Austria}{joel.rybicki@ist.ac.at}{https://orcid.org/0000-0002-6432-6646}{}
\author{Stefan Schmid}{TU Berlin, Germany}{stefan.schmid@tu-berlin.de}{}{Austrian Science Fund (FWF) project DELTA (grant agreement No I 5025-N).}
\author{Jan Studený}{Aalto University, Finland}{jan.studeny@aalto.fi}{}{}
\author{Jukka Suomela}{Aalto University, Finland}{jukka.suomela@aalto.fi}{https://orcid.org/0000-0001-6117-8089}{}
\author{Jara Uitto}{Aalto University, Finland}{jara.uitto@aalto.fi}{https://orcid.org/0000-0002-5179-5056}{}

\authorrunning{A.\ Balliu et al.}
\Copyright{Alkida Balliu, Janne H.\ Korhonen, Fabian Kuhn, Henrik Lievonen, Dennis Olivetti, Shreyas Pai, Ami Paz, Joel Rybicki, Stefan Schmid, Jan Studený, Jukka Suomela, and Jara Uitto}



\nolinenumbers
\hideLIPIcs

\theoremstyle{definition}

\DeclareMathOperator{\poly}{poly}
\newcommand{\algoa}{\mathcal{A}}
\newcommand{\algob}{\mathcal{B}}
\newcommand{\algoc}{\mathcal{C}}
\newcommand{\A}{\mathcal{A}}
\newcommand{\slocal}{\ensuremath{\mathsf{SLOCAL}}\xspace}
\newcommand{\local}{\ensuremath{\mathsf{LOCAL}}\xspace}

\begin{document}
\maketitle

\begin{abstract}
The sinkless orientation problem plays a key role in understanding the foundations of distributed computing.
The problem can be used to separate two fundamental models of distributed graph algorithms, \local and \slocal:
the locality of sinkless orientation is $\Omega(\log n)$ in the deterministic \local model and $O(\log \log n)$ in the deterministic \slocal model.
Both of these results are known by prior work, but here we give new simple, self-contained proofs for them.
\end{abstract}

\newpage

\section{Introduction}

One of the fundamental challenges in the study of graph algorithms concerns the understanding of the \emph{locality} of the considered graph problem: given a node in the middle of a large graph, \emph{how far} do we need to see around that node to choose its output? For example, if we are interested in the graph coloring problem, how far do we need to see around a node before we can choose its color, so that the end result is a globally consistent coloring?

The notion of locality plays a particularly important role in characterizing the \emph{distributed complexity} of graph problems~\cite{Linial1992,Peleg2000}---for example, problems that are local can be solved in a distributed setting with a small number of communication rounds. The past decade has seen a successful research program~\cite{Brandt2016,chang16exponential,Ghaffari2017,Ghaffari2018,Balliu2018disc,Balliu2018stoc,Foerster2019,rozhon2020polylogarithmic} contributing to our systematic understanding of the fundamental interplay between locality, randomness, and the computational power of different models of distributed graph algorithms. 

In this work, we give a new, simple proof for one of the key results in this area: the sinkless orientation problem gives an exponential separation between the \local and \slocal models of computing. The standard approach for proving this result relies on fairly heavy-weight machinery, whereas our new proof is elementary and self-contained.

\subsection{Sinkless Orientation}

In the \emph{sinkless orientation} problem, we are given an undirected graph as input, and the task is to orient all edges so that all nodes of degree at least $3$ have at least one outgoing edge (i.e., they are not sinks). Here are some examples of valid solutions:
\begin{center}
\includegraphics[page=1]{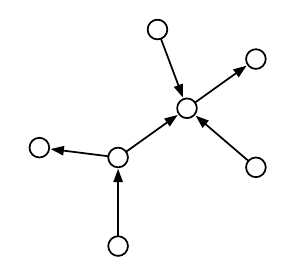}%
\hspace{1cm}%
\includegraphics[page=2]{figs.pdf}
\end{center}

A sinkless orientation always exists in any graph and it is easy to find given a global view of the input graph: Process each connected component separately. If the component is a tree, we can choose a leaf node $v$ and orient everything towards $v$.
Otherwise there is a cycle $C$, and we can then orient $C$ in a consistent direction and orient all other edges towards $C$, breaking ties arbitrarily.
The following figure illustrates both of these cases:
\begin{center}
\includegraphics[page=3]{figs.pdf}%
\hspace{1cm}%
\includegraphics[page=4]{figs.pdf}
\end{center}

However, this simple algorithm is inherently global---the orientation of a given edge depends on information arbitrarily far from it. The key question that was first explicitly asked in 2016 \cite{Brandt2016} regards the locality of the sinkless orientation problem: can one come up with a rule that always results in a sinkless orientation such that each edge is oriented based on the information that is within its radius-$T(n)$ neighborhood, where $T(n)$ is some sublinear function of the number of nodes $n$, or ideally a constant function independent of $n$?

\subsection{LOCAL and SLOCAL Models}

We consider the \local and \slocal models of (distributed) graph algorithms. For both models, the setting is as follows. We are given an input graph $G = (V,E)$ on $n$ nodes and the goal is to compute a sinkless orientation on $G$. Each node $v \in V$ has to produce a \emph{local output}, in our case an orientation of all edges incident to $v$. The local output of $v$ is determined by an algorithm that has access to the information available in $G$ within distance $T(n)$ from $v$, where $T \colon \mathbb N \to \mathbb N$ is a function of the size of the input graph. The key difference between the \local and \slocal models is in the way nodes are processed (see \cref{fig:models}):

\begin{description}
    \item[Deterministic \local model:] Each node $v \in V$ chooses its local output \emph{simultaneously in parallel} based on the information available within distance $T(n)$ from $v$. That is, each node $v$ maps its radius-$T(n)$ neighborhood in $G$ to an output value.

    \item[Deterministic \slocal model:] The nodes are processed \emph{sequentially} in some arbitrary order chosen by an adversary. When node $v \in V$ is processed, it chooses its internal state and output based on the information available within distance $T(n)$ in the graph. This information also includes the internal states of the nodes processed \emph{before} node~$v$.
\end{description}

\begin{figure}
\centering
\includegraphics[page=5]{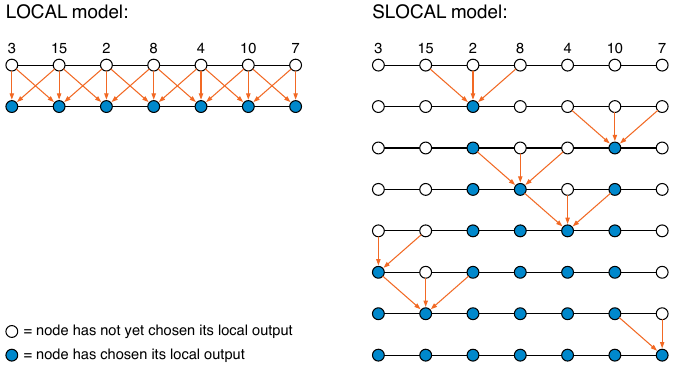}
\caption{Information flow in the \local and \slocal models; in this example we consider locality $T(n) = 1$, i.e., each node chooses its output based on its radius-$1$ neighborhoods. The input graph here is a path with $7$ nodes; the nodes are labeled by the adversary with unique identifiers. In the \local model, all nodes choose their final state simultaneously in parallel, while in the \slocal model, the nodes make their choices sequentially, in some order chosen by the adversary.}\label{fig:models}
\end{figure}

In both models, we assume that the number of nodes $n$ is known and that the nodes of the graph are labeled with unique identifiers from $1$ to $\poly(n)$; this is particularly important for the \local model so that one can break symmetry.

The \emph{locality of a graph problem $\Pi$} is the smallest distance sufficient to solving $\Pi$ in the given model. That is, if $\Pi$ has locality $T(n)$, then there is an algorithm $\A$ that (1) uses only information available within distance $T(n)$ to compute the output of any node $v$ and (2) the output is correct on any input graph, on any choice of unique identifiers, and, in the \slocal model, on any processing order of nodes.

\subsection{Is SLOCAL Any Stronger Than LOCAL?}

For any problem $\Pi$, the locality of $\Pi$ in the \slocal model is at most as large as the locality of $\Pi$ in the \local model: if we can solve $\Pi$ in the \local model so that each node makes a choice based on its radius-$T(n)$ neighborhood, we can do the same in the \slocal model: the algorithm can just ignore the internal states of nodes. However, the key question is if the locality of a problem can be much smaller in the \slocal model than in the \local model.

The answer may seem obvious. For example, consider the problem of coloring a path with $3$ colors:
\begin{itemize}
    \item In the \slocal model, we can solve this problem with locality $T(n) = 1$: each node can greedily pick a free color that is not yet used by any of its neighbors.
    \item In the \local model, there is no constant-locality algorithm---this is the seminal result by Linial~\cite{Linial1992}, and there is a simple version of the proof in \cite{Laurinharju2014}.
\end{itemize}
However, this problem only gives a slightly super-constant separation between the models: there is an algorithm with locality $T(n) = O(\log^* n)$ that solves the problem in the \local model by making clever use of the unique identifiers \cite{cole86deterministic,Barenboim2013}; here $\log^*$ is the inverse of a power tower, i.e., a very slowly-growing function. More generally, it turns out that any problem that can be solved in the \slocal model with locality $O(1)$ can be solved in the \local model with locality $O(\log^* n)$ \cite{Ghaffari2018}.

If we could always turn any \slocal algorithms into \local algorithms with only $O(\log^* n)$ overhead in locality, this would be great news for the designers of distributed algorithms: it is often much easier to reason about sequential \slocal algorithms than about parallel \local algorithms. However, sinkless orientation shows that this is not the case.

\subsection{Sinkless Orientation Separates LOCAL and SLOCAL}

Sinkless orientation can be used to prove a strong separation between deterministic \local and deterministic \slocal.
By prior work \cite{Brandt2016,chang16exponential,ghaffari17distributed,Ghaffari2018}, we know that:
\begin{restatable}{theorem}{thmmainlower}\label{thm:main-lower}
The locality of the sinkless orientation problem in the deterministic \local model is $\Omega(\log n)$.
\end{restatable}
\begin{restatable}{theorem}{thmmainupper}\label{thm:main-upper}
The locality of the sinkless orientation problem in the deterministic \slocal model is $O(\log \log n)$.
\end{restatable}

Unfortunately, even though these results play a key role in understanding the landscape of models of distributed computing (see \cref{sec:discussion} for the broader context), there have not been simple proofs for either of these results. While the theorems are related to \emph{deterministic} models, the prior proofs of \cref{thm:main-lower,thm:main-upper} take a detour through \emph{randomized} models, and apply fairly heavyweight machinery:
\begin{itemize}
    \item The prior proof of \cref{thm:main-lower} first shows that the locality is $\Omega(\log \log n)$ in the randomized \local model \cite{Brandt2016}; this requires a careful analysis of how the local failure probability of a randomized algorithm behaves in the so-called \emph{round elimination} technique. Then we can conclude that the locality is also $\Omega(\log \log n)$ in the deterministic \local model. Finally, we can apply a general \emph{gap result} to extend the lower bound to $\Omega(\log n)$ \cite{chang16exponential}.
    \item The prior proof of \cref{thm:main-upper} first constructs an algorithm with locality $O(\log \log n)$ in the randomized \local model \cite{ghaffari17distributed}; here one can use the so-called \emph{shattering} technique, and argue that after the randomized shattering phase, which orients only some edges of the graph, the connected components of what remains to be processed are small enough so that even if one solves them deterministically, locality of $O(\log \log n)$ suffices. Then one can apply a generic derandomization result that enables the simulation of randomized \local with deterministic \slocal \cite{Ghaffari2018}, and the result follows.
\end{itemize}

\subsection{Contributions and Key Ideas}

We provide new short, elementary, and entirely self-contained proofs for \cref{thm:main-lower,thm:main-upper}.

\subparagraph*{The lower bound.}
To obtain \cref{thm:main-lower}, we prove a stronger lower bound: it turns out to be convenient to work in the \emph{supported} version of the \local model~\cite{Schmid2013}. In the supported \local model, there is a \emph{support graph} $G$ that is known to all nodes in advance, and the input graph $H$ is a subgraph of $G$.
The fact that $G$ is globally known makes the supported model stronger than the usual \local model (e.g.\ graph coloring is trivial, as a proper vertex coloring of $G$ gives a proper vertex coloring of~$H$). We show that the locality of the sinkless orientation problem is $\Omega(\log n)$ in the supported \local model, which implies the same lower bound in the \local model.

\subparagraph{The upper bound.}
To prove \cref{thm:main-upper}, we introduce the \emph{high-degree sinkless orientation problem}, in which we only care that nodes with high degree are not sinks.
This problem is trivial to solve in the \slocal model.
We then provide an \slocal algorithm which constructs a virtual graph on top of the actual graph and solves the high-degree sinkless orientation problem on the virtual graph.
The algorithm then lowers the solution on the virtual graph to a solution for the ordinary sinkless orientation problem on the original graph.

\section{Sinkless Orientation Has Locality \texorpdfstring{\boldmath $\Omega(\log n)$}{Ω(log n)} in LOCAL}

In this section, we show that the locality of the sinkless orientation problem in the deterministic \local model is $\Omega(\log n)$. 
We in fact prove the lower bound in the stronger \emph{supported \local model}. In this variant of the \local model, there is a globally known \emph{support graph} $G = (V,E)$ with known assignment of unique identifiers, and the input is a subgraph $H$ of $G$. In an algorithm with locality $T(n)$, each node $v$ receives as input \emph{the entire structure} of the support graph $G$, including all the unique identifiers, and information about which edges in its radius-$T(n)$ neighborhood in $G$ belong to $H$; we refer to edges of $H$ as input edges. In our case, we would like to find a sinkless orientation in the input graph $H$.

\subparagraph*{Roadmap.} For technical convenience, we prove the result in a stronger \emph{bipartite} version of the supported \local model. The lower bound in this setting then implies lower bounds for (non-bipartite) \local and supported \local, by observing that algorithms from a weaker model can be translated to the stronger models with no overhead in locality.

The overall structure of our lower bound proof is as follows. We fix a bipartite $5$-regular graph $G$ with girth $\Omega(\log n)$, and an assignment of unique identifiers on $G$. We then show that in bipartite supported \local, any algorithm that solves sinkless orientation, even with the promise that the support graph is $G$, has locality $\Omega(\log n)$.

The proof has two main steps. First, we give a round elimination lemma showing that any sinkless orientation algorithm with locality $T$ on $H$ can be converted into an algorithm with locality $T-1$, if $T$ is sufficiently less than the girth of $G$. By iterating this lemma, we can turn an algorithm with locality $T$ into an algorithm with locality $0$. Second, we show that no such trivial algorithm with locality $0$ can exist, implying that any algorithm requires $\Omega(\log n)$ locality.

\subsection{Setup}

\subparagraph*{Bipartite model.} In bipartite supported \local, we are given a promise that the support graph $G$ is bipartite, and a $2$-coloring is given to the nodes as an input; we refer to the two colors as black and white.
In the bipartite model, we consider either the black or white nodes to be \emph{active}, and the other color to be \emph{passive}. All nodes of the graph run an algorithm as per the supported \local model; upon termination of the algorithm, the active nodes produce an output, and the passive nodes output nothing. The outputs of the active nodes must form a globally valid solution; in particular, in sinkless orientation, the outputs of the active nodes already orient all edges, and neither active or passive nodes can be sinks.

\subparagraph*{Sinkless orientation in bipartite model.}
We encode sinkless orientation in the bipartite supported \local model as follows.
Each active node outputs, for each incident input edge,
one label from the alphabet $\Sigma = \{ \mathsf{O}, \mathsf{I}\}$.
The edge-output $\mathsf{O}$ indicates that the edge is outgoing from the active node, and the edge-output $\mathsf{I}$
indicates it is incoming to the active node. An output is correct if for each active node of degree at least $3$, there is at least one output $\mathsf{O}$ on an incident input edge, and for each passive node of degree at least $3$, there is at least one output $\mathsf{I}$ on an incident input edge.
	The labels $\mathsf{O}, \mathsf{I}$ represent orientation w.r.t.\ the active node, and we require each active node to have at least one $\mathsf{O}$ label (indicating an outgoing edge), and each passive node to have at least one $\mathsf{I}$ label (indicating an edge incoming to an active neighbor, thus outgoing from the passive node we consider). Hence, any solution on general graphs can immediately be translated to a solution in the bipartite~model.

	In more detail, consider a sinkless orientation algorithm $\A$ with locality $T$ in the (supported) \local model, with some reasonable output encoding. To turn this into a bipartite (supported) \local algorithm, one first runs algorithm $\A$ in the bipartite model---this requires no modifications, as computation in the bipartite model is done exactly as in the original. After $\A$ has terminated, (1) the passive nodes discard the output of $\A$ and output nothing, and (2) the active nodes inspect the output of $\A$, and output $\mathsf{I}$ for each incident edge directed towards them, and $\mathsf{O}$ for each edge directed away from them in the output of $\A$. Since $\A$ is a sinkless orientation algorithm, these outputs also guarantee that each passive node has one edge with output $\mathsf{I}$ incident to it. In particular, it follows that lower bounds for bipartite algorithms are also lower bounds for the standard models.

\subsection{Step One: Round Elimination}\label{sec:restep}

\begin{lemma}\label{lemma:re}
Let $G$ be a fixed $5$-regular bipartite graph with girth $g$, and fixed unique identifiers and $2$-coloring of the nodes.
Let $0 < T < g/2$, and assume there is an algorithm $\A_T$ that solves sinkless orientation on support graph $G$ with locality $T$. Then there is an algorithm $\A_{T-1}$ that solves sinkless orientation on support graph $G$ with locality $T-1$.
\end{lemma}

\begin{proof}
The proof proceeds by the standard round elimination strategy. Let us assume without loss of generality that black nodes are active in $\A_T$. For a non-negative integer $t$ and any node $v \in V$, let us denote by $B(v,t)$ the nodes within distance $t$ from the node $v$ in the graph $G$.
We construct an algorithm $\A_{T-1}$ where white nodes are active. In algorithm $\A_{T-1}$, each white node $u \in V$ performs the following steps:
\begin{enumerate}
    \item Node $u$ gathers the inputs in its $(T-1)$-radius neighborhood $B(u,T-1)$.
    \item For each neighbor $v$ of $u$, the node $u$ enumerates all possible input graphs $H'$ on $B(v,T)$ that are compatible with the actual input graph $H$ on $B(u,T-1)$. For each such $H'$, $u$ simulates $\A_T$ to compute what output $v$ would output on the edge $\{ u, v \}$ under input~$H'$. Let $S(u,v)$ denote the set of all possible outputs obtained for edge $\{ u, v\}$ this way.
    \item
    If $S(u,v)=\{ \mathsf{I} \}$ then $u$  outputs $\mathsf{O}$ on $\{u,v\}$, and otherwise it outputs $\mathsf{I}$ on it.
\end{enumerate}
We now prove $\A_{T-1}$ produces a valid solution for sinkless orientation.

Consider a white node $u$, and two of its neighbors $v$, $v'$ in $H$.
Since $T < g/2$, we have $B(v, T) \cap B(v',T) = B(u,T-1)$, and thus the inputs in $B(v,T) \setminus B(u,T-1)$ do not affect the output of $v'$ in $\A_T$, and likewise the inputs in $B(v',T) \setminus B(u,T-1)$ do not affect the output of $v$ in $\A_T$.
Thus, any combination of $\mathsf{L}_v \in S(u,v)$ and $\mathsf{L}_{v'} \in S(u,v')$ may occur as an output:
for any such $\mathsf{L}_v$ and $\mathsf{L}_{v'}$, there is an input graph such that $v$ in $\A_T$ outputs $\mathsf{L}_v$ for the edge $\{ u, v \}$, and $v'$ outputs $\mathsf{L}_{v'}$ for the edge $\{ u, v' \}$.

Let $u$ be a white node of degree at least $3$ in $H$ with neighbors $N(u)$ in $H$.
By the above argument, for any choice of one $\mathsf{L}_v \in S(u,v)$ for each neighbor $v\in N(u)$ of $u$, there is an input graph on which $\A_T$ outputs $\mathsf{L}_v$ for the edge $\{ u, v \}$.
If each $S(u,v)$ contains $\mathsf{O}$, there would be an input graph
on which $\A_T$ outputs $\mathsf{O}$ for all incident input edges of $u$, rendering it incorrect.
Hence, at least one neighbor $v'$ of $u$ satisfies $S(u,v') = \{ \mathsf{I} \}$, and in $\A_{T-1}$ where $u$ is active, $u$ outputs $\mathsf{O}$ on the edge $\{u,v'\}$.

On the other hand, consider black node $v$ of degree at least $3$ with neighbors $N(v)$ in $H$.
On the true input $H$ node $v$ in $\A_T$ will output $\mathsf{O}$ on an incident edge $\{ v, u \}$, for some $u \in N(v)$.
In $\A_{T-1}$, the node $u$ will consider the input $H$ on $B(v,T)$ (among other inputs), so we have $\mathsf{O} \in S(u, v)$.
Thus, in $\A_{T-1}$ the node $u$ will output $\mathsf{I}$ on $\{ v, u \}$, and $v$ has an incident edge labeled $\mathsf{I}$ as desired.
\end{proof}

\subsection{Step Two: There Exists No Algorithm with Locality 0}

\begin{lemma}\label{lemma:0round}
Let $G$ be a fixed $5$-regular bipartite graph with girth $g$, and assume unique identifiers and $2$-coloring on $G$ are fixed. The locality of sinkless orientation in bipartite supported \local on support graph $G$ is greater than $0$.
\end{lemma}

\begin{proof}
Assume for contradiction that there is an algorithm $\A_0$ with locality $0$ and black nodes as active. Label each edge $e$ of $G$ by the set of all outputs $\A_0$ can output for $e$ when $e$ is part of the input. For any black node $v$, there must be at least three edges labeled with either $\{ \mathsf{O} \}$ or $\{ \mathsf{O}, \mathsf{I} \}$, as otherwise, for some input $v$ would have exactly three incident input edges on which it would output $\mathsf{I}$.
Since every edge is incident to exactly one black node, at most $2/5$ of the edges are labeled $\{ \mathsf{I} \}$.
Hence, there is a white node $u$ such that $u$ is incident to at least three edges $\{ u, v_1 \}$, $\{ u, v_2 \}$ and $\{ u, v_3 \}$ labeled with either $\{ \mathsf{O} \}$ or $\{ \mathsf{O}, \mathsf{I} \}$. Now consider an input where these three edges are the only input edges incident to $u$. Since the output of each node $v_i$ depends only on its incident input edges, we can select for each $v_i$ an input where $v_i$ outputs $\mathsf{O}$ for edge $\{ u, v_i \}$. Moreover, since $\A_0$ is an algorithm with locality $0$ and nodes $v_1$, $v_2$ and $v_3$ are not neighbors, we can do this for all of them simultaneously. Thus, there exists an input where $\A_0$ outputs $\mathsf{O}$ on all incident input edges of the passive node $u$, a contradiction.
\end{proof}

\subsection{Putting Things Together}

\begin{theorem}\label{thm:supported-lb}
The locality of the sinkless orientation problem in the deterministic supported \local model is $T(n) = \Omega(\log n)$.
\end{theorem}
\begin{proof}
Let $G$ be a bipartite 5-regular graph with girth $g = \Omega(\log n)$.
Observe that we can obtain one e.g.~by taking the bipartite double cover of any $5$-regular graph of girth $\Omega(\log n)$, which are known to exist (see e.g.,~\cite[Ch.\ 3]{bollobas2004extremal}).

Assume that there is a supported \local algorithm $\A_T$ that solves sinkless orientation with locality $T < g/2$ on support graph $G$. This implies that there is a bipartite supported \local algorithm for sinkless orientation on $G$ running in time $T$. By repeated application of \cref{lemma:re}, there is a sequence of bipartite supported \local algorithms
$\A_T, \A_{T-1}, \dotsc, \A_1, \A_0,$
where algorithm $\A_i$ solves sinkless orientation with locality $i$.

In particular, $\A_0$ solves sinkless orientation with locality $0$. By \cref{lemma:0round}, this is impossible, so algorithm $\A_T$ cannot exist.
\end{proof}

\thmmainlower*
\begin{proof}
Any \local algorithm $\A$ with locality $T(n)$ can be simulated in supported \local with locality $T(n)$ by ignoring non-input edges: simply run $\A$ on the input graph $H$. Thus, the claim follows immediately from \cref{thm:supported-lb}.
\end{proof}

\section{Sinkless Orientation Has Locality \texorpdfstring{\boldmath $O(\log\log n)$}{O(log log n)} in SLOCAL}

We now show that the locality of the sinkless orientation problem in the deterministic \slocal model is $O(\log\log n)$.

\subparagraph*{Roadmap.}
As the first step, we consider a variant of the sinkless orientation problem called \emph{high-degree sinkless orientation}, where only nodes with degree $\Omega(\log n)$ are required not to be sinks. We show that this problem can be solved with a simple greedy algorithm that processes edges one at a time, and this algorithm can be implemented in \slocal with locality $O(1)$.
As the second step, we show how to reduce the general case to high-degree sinkless orientation.

As the high-level idea, we compute a clustering of the nodes using an $\Theta(\log \log n)$-independent set, and solve high-degree sinkless orientation on the graph formed by the clusters and edges between the clusters. We can then orient edges inside each cluster independently without creating sinks; high-degree clusters will already have one outgoing edge oriented away from the cluster, and low-degree clusters contain either a node of degree at most two or a cycle. Moreover, this idea can be implemented as an \slocal algorithm with locality $O(\log \log n)$.

\subsection{Step One: High-Degree Sinkless Orientation}\label{sec:high-degree-so}

\newcommand{\degreebound}{\ensuremath{\lfloor\log_2 n\rfloor+1}}
\newcommand{\localitybound}{\ensuremath{\lceil\log_2(\degreebound)\rceil}}

\emph{The high-degree sinkless orientation problem} is a variation of the sinkless orientation problem in which we only care that nodes with degree of at least $\degreebound$ are not sinks; we call such nodes \emph{high-degree nodes}. For technical purposes, we will assume in this section that the input graph $G = (V,E)$ is a multigraph.

\subparagraph*{Greedy algorithm.}

We describe a greedy algorithm $\algoa$ for solving high-degree sinkless orientation that orients the edges of the input multigraph $G =(V,E)$ one at a time.
During the execution of $\algoa$, we say that a node $v \in V$ is \emph{satisfied} if it either is not a high-degree node or at least one incident edge has been oriented away from $v$; otherwise, $v$ is \emph{unsatisfied}. Algorithm $\algoa$ processes each edge $e = \{ u, v \}$ using the following rules:
\begin{enumerate}
	\item If either $u$ or $v$ is already satisfied, the algorithm orients the edge towards the satisfied node, breaking ties arbitrarily.
	\item\label[rule]{rule:unsatisfied-combination} Otherwise, both $u$ and $v$ are unsatisfied high-degree nodes. We orient the edge towards the node which has fewer adjacent edges already processed, breaking ties arbitrarily.
\end{enumerate}

\begin{lemma}
  Algorithm $\algoa$ produces a valid solution to high-degree sinkless orientation.
\end{lemma}
\begin{proof} At each step of the execution of $\algoa$, consider the connected components formed by the edges that have been processed \emph{by \cref{rule:unsatisfied-combination}} up to current step. We want to show that the following invariant holds:
	if an unsatisfied node $v \in V$ has $b$ edges oriented towards $v$, then the current connected component has at least $2^b$ nodes. This suffices to prove the theorem, as any unsatisfied node after the termination of the algorithm would need to be part of a component containing at least $2^{\degreebound} > n$ nodes, and therefore the component needs to be larger than the whole graph, a contradiction.
The invariant trivially holds before any edges have been processed, as each node has $0$ edges directed towards them and each current component consists of one node. The invariant also trivially remains true after any step where we process an edge $e = \{ u, v\}$, where $u$ or $v$ is satisfied.
Consider now the case where the algorithm processes an edge $e = \{ u, v \}$ with both $u$ and $v$ unsatisfied, and let the number of processed edges incident to $u$ and $v$ be $b_u$ and $b_v$, respectively.   Without loss of generality, we may assume that $b_v \leq b_u$ and that $\algoa$ orients $e$ towards $v$. Node $u$ is now satisfied, and node $v$ has indegree $b_v + 1$. Since the invariant held before this step, we have that the new connected component containing $v$ now has size at least
$
    2^{b_v} + 2^{b_u} \ge 2^{b_v} + 2^{b_v} = 2^{b_v + 1},
$ 
implying that the invariant holds.
\end{proof}

\subsection{Step Two: Sinkless Orientation on General Graphs}

We start by describing our algorithm for sinkless orientation on general graphs in three steps.
Each one of these steps can be implemented in the \slocal model with locality $O(\log \log n)$, assuming that the output from previous steps is available at the nodes.
We defer the proof that these steps can be combined into a single-step \slocal algorithm with locality $O(\log \log n)$ until the end of the section.
In the following, let $T = \localitybound$.

\subparagraph*{Clustering.} We construct a \emph{clustering} of the input graph $G = (V,E)$ by computing a maximal independent set $I$ in graph $G^{2T + 1}$ and assigning each node to the cluster of the closest independent set node $v \in I$, breaking ties arbitrarily. For node $v \in I$, we denote by $C_v$ the cluster corresponding to $v$. We say that an edge $e \in E$ is an \emph{inter-cluster} edge if its endpoints are in  different clusters, and an \emph{intra-cluster} edge otherwise.

We note that the radii of the clusters are bounded:
All nodes in the radius-$T$ neighborhood of node $v \in I$ belong to cluster $C_v$.
On the other hand, for every node $u$ in $C_v$, the distance between $u$ and $v$ is at most $2T$.

We now define a virtual \emph{cluster graph} with a node for each cluster $C_v$ for $v \in I$, and adding an edge between $C_v$ and $C_u$ for every edge of the original graph that connects a node in $C_v$ to $C_u$, preserving duplicate edges. That is, the cluster graph can be a multigraph.

Finally, we note that this clustering can be done with locality $O(T)$ by a simple greedy \slocal algorithm:
When the algorithm processes node $v$, it checks whether there are any other nodes belonging to set $I$ in the radius-$(2T+1)$ neighborhood of $v$.
If there are, then node $v$ belongs to the cluster of the nearest such node, and otherwise the algorithm adds node $v$ to set $I$.

\subparagraph*{Orienting inter-cluster edges.} As the next step, we compute a high-degree orientation of the cluster graph, using the greedy algorithm of \cref{sec:high-degree-so}. Since there is a one-to-one correspondence between the edges of the cluster graph and edges between the clusters in the input graph $G$, this naturally induces an orientation of inter-cluster edges in $G$. We observe that under this partial orientation, any cluster $C_v$ with degree at least $\degreebound$ has at least one edge oriented away from $C_v$, as the number of nodes in the cluster graph is at most $n$ and thus $C_v$ counts as a high-degree node.

Again, this step can be implemented in the \slocal model with locality $O(T)$:
When the algorithm processes a node that is adjacent to an unprocessed inter-cluster edge, it can fully see both of the clusters in its radius-$O(T)$ neighborhood, including the direction of inter-cluster edges that have been previously processed.

\subparagraph*{Orienting intra-cluster edges.} Finally, we show that given the orientation of inter-cluster edges as above, the intra-cluster edges can be oriented without creating any sinks of degree $3$ or higher. We have two cases to consider:
\begin{itemize}
	\item \emph{High-degree} clusters with at least $\degreebound$ inter-cluster edges have at least one outgoing inter-cluster edge.
	\item \emph{Low-degree} clusters with less than $\degreebound$ inter-cluster edges may have all inter-cluster edges directed towards the cluster.
\end{itemize}
We show that in both cases, it is possible to compute an orientation of intra-cluster edges based on the internal structure of the cluster and the orientation of the boundary edges so that no node of degree at least $3$ is a sink.

For a high-degree cluster $C$, we know that there is a node $v \in C$ with an outgoing inter-cluster edge. In this case, picking an arbitrary spanning tree for $C$, orienting its edges towards $v$, and orienting remaining edges arbitrarily clearly suffices.

For a low-degree cluster $C$, we first observe that $C$ cannot be locally tree-like:
\begin{lemma}
  A low-degree cluster $C_v$ contains either a cycle or a node with degree 1 or 2.
\end{lemma}
\begin{proof}
  Assume for contradiction that the cluster does not contain a cycle and that the degree of every node is at least 3.
  Recall that all nodes within distance $T$ of the cluster center $v$ are contained in $C_v$, and thus there are at least
  \begin{equation*}
    3 \cdot 2^{T-1} > 2^{T} = 2^{\localitybound} \ge \degreebound
  \end{equation*}
  nodes in $C_v$ at distance $T$ from $v$.
  Moreover, it follows that there are at least this many edges on the boundary of the cluster, and thus the cluster is adjacent to at least $\degreebound$ inter-cluster edges, a contradiction
\end{proof}

If the cluster contains a cycle, then we can orient that cycle in a consistent manner, and orient the rest of the edges towards the cycle.
Otherwise the cluster contains a node with degree 1 or 2, in which case we orient all edges towards that node.

As the radius of each cluster is bounded by $2T$, every node can see the whole cluster it belongs to within its radius-$O(T)$ neighborhood.
Therefore the algorithm can orient the intra-cluster edges in a consistent manner with locality $O(T)$.

\subparagraph*{\boldmath Composability of \slocal algorithms.}
To conclude the description of our algorithm, we to show that one can compose a multiple-step \slocal algorithm into a one-step \slocal algorithm.
This is a well-known result \cite{Ghaffari2017}, but we include a short proof for completeness.
\begin{lemma}
  \label{lemma:slocal-composition}
  Let $\algoa$ and $\algob$ be \slocal algorithms with localities $T_\algoa$ and $T_\algob$, respectively, and let $\algob$ depend on the output of $\algoa$.
  Then there exists an \slocal algorithm $\algoc$ with locality $T_\algoa + 2T_\algob$ that solves the same problem as $\algob$ without dependency on the output of $\algoa$.
\end{lemma}
\begin{proof}
  To compute the output of algorithm $\algob$ for node $v$, algorithm $\algoc$ needs to first compute the output of $\algoa$ in the radius-$T_\algob$ neighborhoods of $v$.
  The challenge here is that $\algoc$ cannot just recompute the output of $\algoa$ every time from scratch as the output may depend on previously processed nodes in the neighborhood.
  To enable this simulation, we allow $\algoc$ to store the output of $\algoa$ for node $u$ at some other node in the radius-$T_\algob$ neighborhood of $u$.

  Algorithm $\algoc$ starts the processing of node $v$ by collecting the output of $\algoa$ in the radius-$T_\algob$ neighborhood of $v$.
  As this output can be stored within distance $T_\algob$ from the actual node, algorithm $\algoc$ requires $2T_\algob$ locality to do this.
  For every node $u$ in the radius-$T_\algob$ neighborhood of $v$ that does not have output for $\algoa$ available, algorithm $\algoc$ collects the radius-$T_\algoa$ neighborhood of $u$ and computes the output for $\algoa$; this can be done with locality $T_\algoa + 2T_\algob$.
  The algorithm then stores the output of node $u$ at node $v$, so that it will not be recomputed later.
  Algorithm $\algoc$ now knows the output of $\algoa$ for all nodes in the radius-$T_\algob$ neighborhood of $v$.
  Therefore it can directly use $\algob$ to compute the output for $v$.
\end{proof}

\thmmainupper*
\begin{proof}
  We can apply \cref{lemma:slocal-composition} twice to combine the three-step \slocal algorithm we described above into a one-step \slocal algorithm.
  As each of the three steps has locality $O(T)$, the final algorithm has locality $O(T) = O(\log \log n)$, completing the proof.
\end{proof}

\section{Discussion and Broader Context}\label{sec:discussion}

In this work, we have presented simple, self-contained proofs of \cref{thm:main-upper,thm:main-lower}, which show that the locality of the sinkless orientation problem in the \local model is exponentially larger than in the \slocal model. We will now briefly discuss the broader context and the role of the sinkless orientation problem and the \slocal model in understanding the foundations of distributed computing.

\subparagraph{Complexity of distributed sinkless orientation.} While we present in this works a lower bound in the \local model and an upper bound in the \slocal model, we note that locality of sinkless orientation is fully understood in these models:
\begin{itemize}
	\item In deterministic \local model, sinkless orientation has locality $\Theta(\log n)$~\cite{Brandt2016,chang16exponential,ghaffari17distributed}.
	\item In randomized \local model, sinkless orientation has locality $\Theta(\log \log n)$~\cite{Brandt2016,ghaffari17distributed}.
	\item In deterministic \slocal model, sinkless orientation has locality $\Theta(\log \log n)$~\cite{ghaffari17distributed,Ghaffari2018}.
	\item In randomized \slocal model, sinkless orientation has locality $\Theta(\log \log \log n)$~\cite{ghaffari17distributed,Ghaffari2018}.
\end{itemize}

\subparagraph{The role of sinkless orientation in understanding the Lov\'{a}sz Local Lemma.}
The sinkless orientation problem was introduced in \cite{Brandt2016} with the purpose of understanding the locality of the constructive Lov\'asz Local Lemma problem in the distributed setting.

Lov\'asz Local Lemma (LLL) is a classic result in probability theory that can be used to show the existence of various combinatorial objects. For example, one can use LLL to prove that a sinkless orientation exists in any graph \cite{Brandt2016}.

In the distributed setting, the key question is the locality of \emph{constructive, algorithmic Lov\'asz Local Lemma}: given a problem where LLL guarantees the \emph{existence} of a solution, what can we say about the locality of \emph{finding} such a solution (e.g.\ in the \local or \slocal model)? For many interesting problems we can prove that a solution exists by using LLL, and hence a generic way to solve LLL in the distributed setting gives a distributed algorithm for all these problems. Since LLL can be used to find a sinkless orientation, any \emph{lower bound} on the locality of sinkless orientation implies also a lower bound on the locality of general LLL algorithms.

\subparagraph{The role of sinkless orientation in understanding splitting problems.}
Sinkless orientation can be seen as the most relaxed version of the \emph{degree splitting} problem, for which two variants exist, \emph{directed} and \emph{undirected}. The directed variant asks for an orientation of the edges such that each node has roughly the same number of incoming and outgoing edges. The undirected variant asks for a coloring of the edges with red and blue such that each node has roughly the same number of red and blue incident edges. Observe that on bipartite two-colored graphs these two problems are equivalent. It is known \cite{ghaffari17distributed} that efficient algorithms for degree splitting allow us to obtain efficient algorithms for e.g.\ edge coloring, and hence understanding the easiest splitting variant (sinkless orientation) may give insights for understanding the more general case.

\subparagraph{The role of sinkless orientation in understanding round elimination.}
In order to prove the $\Omega(\log \log n)$ rounds lower bound for the sinkless orientation problem, authors of \cite{Brandt2016} used the so-called \emph{round elimination} technique. Since then, this technique has been better understood, and sinkless orientation played a key role in developing this technique, which has been since then used to show lower bounds for many fundamental problems \cite{Balliu2019,BBOrulingset,BBKU22,Brandt2019automatic}. On a high level, the standard way of applying this technique works as follows:
\begin{enumerate}
	\item First, prove a lower bound for deterministic algorithms in a weaker setting, where nodes do not have IDs, using a strategy similar to what we do in \Cref{sec:restep}.
	\item Then, lift this lower bound to a stronger setting, where nodes have no IDs but ran\-dom\-iza\-tion is allowed. This step is quite non-trivial, since it requires one to track how the failure probability evolves when making the algorithm one round faster.
	\item Finally, convert the obtained randomized lower bound into a stronger deterministic lower bound for the \local model, by using non-trivial techniques typically used to prove \emph{gap results} in the \local model.
\end{enumerate}
One of our contributions is to simplify this three step process by showing how to directly handle unique identifiers in a round elimination proof. A similar concept for handling unique IDs, called the ID graph technique, was independently discovered in \cite{Brandt2022}.

\subparagraph{\boldmath The role of \slocal.}
The \slocal model has played a key role in understanding the \local model itself. 
One of the major challenges that we encounter in the \local model is the fact that all nodes act in parallel, and they have to decide their output at the same time. The \slocal model abstracts away this issue, since in this model nodes are processed sequentially. Hence, developing algorithms for the \slocal model may be much easier than developing algorithms for the \local model. Combining this with the fact that, by paying some overhead, we have black box ways to convert \slocal algorithms to \local ones \cite{Ghaffari2018}, this gives us an easier way to design \local algorithms. Moreover, \slocal played an important role for understanding the role of randomness in the \local model. In fact, it has been shown that any randomized \local algorithm can be derandomized by paying an $O(\poly \log n)$ overhead \cite{Ghaffari2018}. This result has been shown by providing an \slocal algorithm as an intermediate step.

\subparagraph{\boldmath Supported \local model.} The supported \local model was originally introduced in the context of \emph{software-defined networks} (SDNs). The underlying idea is that the communication graph $G$ represents the unchanging physical network, and the input graph represents the logical state of the network to which the control plane (here, distributed algorithm) needs to respond to; see reference~\cite{Schmid2013} for more details. However, supported \local have proven to be useful as a purely theoretical model for lower bounds~\cite{Foerster2019, HaeuplerWZ21}. 

\newpage

\bibliography{simple-so}

\begin{thebibliography}{10}

\bibitem{Balliu2019}
Alkida Balliu, Sebastian Brandt, Juho Hirvonen, Dennis Olivetti, Mika{\"{e}}l
  Rabie, and Jukka Suomela.
\newblock {Lower bounds for maximal matchings and maximal independent sets}.
\newblock In {\em Proc. 60th Annual IEEE Symposium on Foundations of Computer
  Science (FOCS 2019)}, 2019.
\newblock \href {https://doi.org/https://doi.org/10.1109/FOCS.2019.00037}
  {\path{doi:https://doi.org/10.1109/FOCS.2019.00037}}.

\bibitem{BBKU22}
Alkida Balliu, Sebastian Brandt, Fabian Kuhn, and Dennis Olivetti.
\newblock Deterministic {$\Delta$}-coloring plays hide-and-seek.
\newblock In {\em Proc.\ 54th Annual ACM SIGACT Symposium on Theory of
  Computing (STOC 2022)}. ACM, 2022.

\bibitem{BBOrulingset}
Alkida Balliu, Sebastian Brandt, and Dennis Olivetti.
\newblock Distributed lower bounds for ruling sets.
\newblock {\em SIAM Journal on Computing}, 51(1):70--115, 2022.
\newblock \href {https://doi.org/10.1137/20M1381770}
  {\path{doi:10.1137/20M1381770}}.

\bibitem{Balliu2018disc}
Alkida Balliu, Sebastian Brandt, Dennis Olivetti, and Jukka Suomela.
\newblock {Almost global problems in the LOCAL model}.
\newblock In {\em Proc. 32nd International Symposium on Distributed Computing
  (DISC 2018)}, 2018.
\newblock \href {https://doi.org/10.4230/LIPIcs.DISC.2018.9}
  {\path{doi:10.4230/LIPIcs.DISC.2018.9}}.

\bibitem{Balliu2018stoc}
Alkida Balliu, Juho Hirvonen, Janne~H Korhonen, Tuomo Lempi{\"{a}}inen, Dennis
  Olivetti, and Jukka Suomela.
\newblock {New classes of distributed time complexity}.
\newblock In {\em Proc. 50th ACM Symposium on Theory of Computing (STOC 2018)},
  2018.
\newblock \href {https://doi.org/10.1145/3188745.3188860}
  {\path{doi:10.1145/3188745.3188860}}.

\bibitem{Barenboim2013}
Leonid Barenboim and Michael Elkin.
\newblock {\em {Distributed Graph Coloring: Fundamentals and Recent
  Developments}}, volume~4.
\newblock 2013.
\newblock \href {https://doi.org/10.2200/S00520ED1V01Y201307DCT011}
  {\path{doi:10.2200/S00520ED1V01Y201307DCT011}}.

\bibitem{bollobas2004extremal}
B{\'e}la Bollob{\'a}s.
\newblock {\em Extremal graph theory}.
\newblock Courier Corporation, 2004.

\bibitem{Brandt2019automatic}
Sebastian Brandt.
\newblock {An Automatic Speedup Theorem for Distributed Problems}.
\newblock In {\em Proc. 38th ACM Symposium on Principles of Distributed
  Computing (PODC 2019)}, 2019.
\newblock \href {https://doi.org/10.1145/3293611.3331611}
  {\path{doi:10.1145/3293611.3331611}}.

\bibitem{Brandt2022}
Sebastian Brandt, Yi-Jun Chang, Jan Greb{\'\i}k, Christoph Grunau, V\'{a}clav
  Rozho\v{n}, and Zolt\'{a}n Vidny\'{a}nszky.
\newblock {Local Problems on Trees from the Perspectives of Distributed
  Algorithms, Finitary Factors, and Descriptive Combinatorics}.
\newblock In {\em 13th Innovations in Theoretical Computer Science Conference
  (ITCS 2022)}, 2022.
\newblock \href {https://doi.org/10.4230/LIPIcs.ITCS.2022.29}
  {\path{doi:10.4230/LIPIcs.ITCS.2022.29}}.

\bibitem{Brandt2016}
Sebastian Brandt, Orr Fischer, Juho Hirvonen, Barbara Keller, Tuomo
  Lempi{\"{a}}inen, Joel Rybicki, Jukka Suomela, and Jara Uitto.
\newblock {A lower bound for the distributed Lov{\'{a}}sz local lemma}.
\newblock In {\em Proc. 48th ACM Symposium on Theory of Computing (STOC 2016)},
  2016.
\newblock \href {https://doi.org/10.1145/2897518.2897570}
  {\path{doi:10.1145/2897518.2897570}}.

\bibitem{chang16exponential}
Yi-Jun Chang, Tsvi Kopelowitz, and Seth Pettie.
\newblock {An Exponential Separation between Randomized and Deterministic
  Complexity in the LOCAL Model}.
\newblock In {\em Proc. 57th IEEE Symposium on Foundations of Computer Science
  (FOCS 2016)}, 2016.
\newblock \href {https://doi.org/10.1109/FOCS.2016.72}
  {\path{doi:10.1109/FOCS.2016.72}}.

\bibitem{cole86deterministic}
Richard Cole and Uzi Vishkin.
\newblock {Deterministic coin tossing with applications to optimal parallel
  list ranking}.
\newblock {\em Information and Control}, 70(1):32--53, 1986.
\newblock \href {https://doi.org/10.1016/S0019-9958(86)80023-7}
  {\path{doi:10.1016/S0019-9958(86)80023-7}}.

\bibitem{Foerster2019}
Klaus-Tycho Foerster, Juho Hirvonen, Stefan Schmid, and Jukka Suomela.
\newblock {On the Power of Preprocessing in Decentralized Network
  Optimization}.
\newblock In {\em Proc. IEEE Conference on Computer Communications (INFOCOM
  2019)}, 2019.
\newblock \href {https://doi.org/10.1109/INFOCOM.2019.8737382}
  {\path{doi:10.1109/INFOCOM.2019.8737382}}.

\bibitem{Ghaffari2018}
Mohsen Ghaffari, David~G Harris, and Fabian Kuhn.
\newblock {On Derandomizing Local Distributed Algorithms}.
\newblock In {\em Proc. 59th IEEE Symposium on Foundations of Computer Science
  (FOCS 2018)}, 2018.
\newblock \href {https://doi.org/10.1109/FOCS.2018.00069}
  {\path{doi:10.1109/FOCS.2018.00069}}.

\bibitem{Ghaffari2017}
Mohsen Ghaffari, Fabian Kuhn, and Yannic Maus.
\newblock {On the complexity of local distributed graph problems}.
\newblock In {\em Proc. 49th ACM SIGACT Symposium on Theory of Computing (STOC
  2017)}, pages 784--797. ACM Press, 2017.
\newblock \href {https://doi.org/10.1145/3055399.3055471}
  {\path{doi:10.1145/3055399.3055471}}.

\bibitem{ghaffari17distributed}
Mohsen Ghaffari and Hsin-Hao Su.
\newblock {Distributed Degree Splitting, Edge Coloring, and Orientations}.
\newblock In {\em Proc. 28th ACM-SIAM Symposium on Discrete Algorithms (SODA
  2017)}, pages 2505--2523. Society for Industrial and Applied Mathematics,
  2017.
\newblock \href {https://doi.org/10.1137/1.9781611974782.166}
  {\path{doi:10.1137/1.9781611974782.166}}.

\bibitem{HaeuplerWZ21}
Bernhard Haeupler, David Wajc, and Goran Zuzic.
\newblock Universally-optimal distributed algorithms for known topologies.
\newblock In {\em {Proc. 53rd Annual ACM SIGACT Symposium on Theory of
  Computing (STOC 2021)}}, pages 1166--1179. {ACM}, 2021.
\newblock \href {https://doi.org/https://doi.org/10.1145/3406325.3451081}
  {\path{doi:https://doi.org/10.1145/3406325.3451081}}.

\bibitem{Laurinharju2014}
Juhana Laurinharju and Jukka Suomela.
\newblock {Brief announcement: Linial's lower bound made easy}.
\newblock In {\em Proc. 33rd ACM SIGACT-SIGOPS Symposium on Principles of
  Distributed Computing (PODC 2014)}, pages 377--378. ACM Press, 2014.
\newblock \href {https://doi.org/10.1145/2611462.2611505}
  {\path{doi:10.1145/2611462.2611505}}.

\bibitem{Linial1992}
Nathan Linial.
\newblock {Locality in Distributed Graph Algorithms}.
\newblock {\em SIAM Journal on Computing}, 21(1):193--201, 1992.
\newblock \href {https://doi.org/10.1137/0221015} {\path{doi:10.1137/0221015}}.

\bibitem{Peleg2000}
David Peleg.
\newblock {\em {Distributed Computing: A Locality-Sensitive Approach}}.
\newblock Society for Industrial and Applied Mathematics, 2000.
\newblock \href {https://doi.org/10.1137/1.9780898719772}
  {\path{doi:10.1137/1.9780898719772}}.

\bibitem{rozhon2020polylogarithmic}
V{\'a}clav Rozho{\v{n}} and Mohsen Ghaffari.
\newblock Polylogarithmic-time deterministic network decomposition and
  distributed derandomization.
\newblock In {\em Proceedings of the 52nd Annual ACM SIGACT Symposium on Theory
  of Computing (STOC 2020)}, pages 350--363, 2020.
\newblock \href {https://doi.org/10.1145/3357713.3384298}
  {\path{doi:10.1145/3357713.3384298}}.

\bibitem{Schmid2013}
Stefan Schmid and Jukka Suomela.
\newblock {Exploiting locality in distributed SDN control}.
\newblock In {\em Proc. ACM SIGCOMM Workshop on Hot Topics in Software Defined
  Networking (HotSDN 2013)}, pages 121--126. ACM Press, 2013.
\newblock \href {https://doi.org/10.1145/2491185.2491198}
  {\path{doi:10.1145/2491185.2491198}}.

\end{thebibliography}

\end{document}